\documentclass[10pt,a4paper]{article}
\usepackage[latin1]{inputenc}
\usepackage{amsmath}%
\usepackage{amsfonts}%
\usepackage{amssymb}%
\usepackage{graphicx}
\usepackage{anysize}
\usepackage{natbib}
\usepackage{lscape}
\usepackage{times}
\usepackage{url}
\usepackage{multicol}
\usepackage{color}
\marginsize{35mm}{30mm}{20mm}{20mm}
\usepackage{subfig}
\usepackage{float}

\newtheorem{theorem}{Theorem}

\newtheorem{corollary}[theorem]{Corollary}

\newtheorem{definition}[theorem]{Definition}

\newtheorem{lemma}[theorem]{Lemma}

\newtheorem{proposition}[theorem]{Proposition}
\newtheorem{remark}[theorem]{Remark}

\newenvironment{proof}[1][Proof]{\textbf{#1.} }{\ \rule{0.5em}{0.5em}}

\title{The explicit Laplace transform for the Wishart process\thanks{We are indebted to José Da Fonseca, Antoine Jacquier, Kyoung-Kuk Kim, Eberhard Mayerhofer, Eckhard Platen, Wolfgang Runggaldier and an anonymous referee for helpful suggestions.}}

\author{
\textrm{Alessandro Gnoatto}\thanks{Dipartimento di Matematica - Universit\`a degli Studi di Padova (Italy) and Mathematisches Institut, LMU M\"unchen (Germany). Email: gnoatto@mathematik.uni-muenchen.de.}
\and\textrm{Martino Grasselli}\thanks{Dipartimento di Matematica - Universit\`a degli Studi di Padova (Italy), ESILV Ecole Sup\'{e}rieure d'Ing\'{e}nieurs L\'{e}onard de Vinci, D\'{e}partement Math\'{e}matiques et Ing\'{e}nierie Financi\`{e}re, Paris La D\'{e}fense (France) and QUANTA FINANZA S.R.L., Venezia (Italy). Email: grassell@math.unipd.it.}
}

\begin{document}

\maketitle
\begin{abstract}
We derive the explicit formula for the joint Laplace transform of the Wishart process and its time integral which extends the original approach of \cite{article_Bru}. We compare our methodology with the alternative results given by the variation of constants method, the linearization of the Matrix Riccati ODE's and the Runge-Kutta algorithm. The new formula turns out to be fast and accurate.
\end{abstract}



\textbf{AMS Class 2010:} 65C30, 60H35, 91B70.\\

\section{Introduction}
In this paper we propose an analytical approach for the computation of the moment generating function for the Wishart process which has been introduced by \cite{article_Bru}, as an extension of square Bessel processes (\cite{article_PY}, \cite{book_RY}) to the matrix case. Wishart processes belong to the class of affine processes and they generalise the notion of positive factor in so far as they are defined on the set of positive semidefinite real $d\times d$ matrices, denoted by $S_d^+$. 
Given a filtered probability space  $(\Omega,\mathcal{F},\mathcal{F}_t,\mathbb{P})$ satisfying the usual assumptions and a $d\times d$ matrix Brownian motion $B$ (i.e. a matrix whose entries are independent Brownian motions  under $\mathbb{P}$), a Wishart process on $S_d^{+}$ is governed by the  SDE
\begin{align}\label{Wis_dyn}
dS_t=\sqrt{S_t}dB_tQ+Q^{\top}dB_t^{\top}\sqrt{S_t}+\left(MS_t+S_tM^{\top}+b\right) dt,\quad S_0\in S^+_d,\quad t\geq0
\end{align}
where $Q\in GL_d$ (the set of invertible real $d\times d$ matrices), $M\in M_d$ (the set of real $d\times d$ matrices) with all eigenvalues on the negative half plane in order to ensure stationarity, and where the matrix $b$ satisfies $b\succeq (d-1)Q^{\top}Q$, that is $b- (d-1)Q^{\top}Q\in S^+_d$. In the literature, the constant drift term is often of the more restrictive form $b=\alpha Q^{\top}Q$, for $\alpha \geq d-1$. In case
the  (Gindikin) real parameter $\alpha$ satisfies  $\alpha\geq d+1$, the process takes values in the interior of $S_d^+$, denoted by $S_d^{++}$, in analogy with the Feller condition for the scalar case. In the dynamics above $\sqrt{S_t}$ denotes the square root in matrix sense. Existence and uniqueness results for the solution of \eqref{Wis_dyn} may be found in \cite{article_Bru} under parametric restrictions and in \cite{article_MPS} in full generality. We denote by $WIS_d(S_0, b, M,Q)$ the law of the Wishart process $(S_t)_{t\geq0}$. The starting point of the analysis was given by considering the square of a matrix Brownian motion $S_t=B_t^\top B_t$, while the  generalization to the particular dynamics \eqref{Wis_dyn} was introduced by looking at squares of matrix Ornstein-Uhlenbeck processes (see \cite{article_Bru}).

Bru proved many interesting properties of this process, like non-collision of the eigenvalues (when $\alpha \geq d+1$ under parametric restrictions) and the additivity property shared with square Bessel processes. Moreover, she computed the Laplace transform of the Wishart process and its integral (the \textit{Matrix Cameron-Martin formula} using her terminology), which plays a central role in the applications: 
\begin{equation}
\mathbb{E}_{S_0}^{\mathbb{P}}\left[\exp\left\{-Tr\left[wS_t+\int_{0}^{t}{vS_sds}\right]\right\}\right],
\label{Laplace}
\end{equation} where $Tr$ denotes the trace operator and $w,v$ are symmetric matrices for which the expression (\ref{Laplace}) is finite. Bru found an explicit formula for (\ref{Laplace}) (formula (4.7) in  \cite{article_Bru}) under the assumption that the symmetric diffusion matrix $Q$ and the mean reversion matrix $M$ commute.

Positive (semi)definite matrices arise in finance in a natural way and the nice analytical properties of affine processes on $S_d^+$ opened the door to new interesting affine models that allow for non trivial correlations among positive factors, a feature which is precluded in classic (linear) state space domains like $\mathbb{R}_{\geq0}^n\times\mathbb{R}^m$ (see \cite{article_DFS}).  Not surprisingly, the last years have witnessed the birth of a whole branch of literature on applications of affine processes on $S_d^+$. The first proposals were formulated in \cite{article_gou01}, \cite{article_gou02}, \cite{article_gou03}, \cite{article_Gou06} both in discrete and continuous time. Applications to multifactor volatility and stochastic correlation can be found in \cite{article_DaFonseca1}, \cite{article_DaFonseca2}, \cite{article_DaFonseca3}, \cite{article_DaFonseca4}, \cite{dafgra11}, \cite{article_BPT}, \cite{BauerleLi2013} and \cite{article_BCT} both in option pricing and portfolio management. These contributions consider the case of continuous path Wishart processes. As far as jump processes on $S_d^+$ are concerned we recall the proposals by \cite{article_BNS02}, \cite{article_MPS02} and \cite{article_PS}. \cite{article_LT} and \cite{article_Cuchiero} consider jump-diffusions models in this class, while \cite{gra08} investigate processes lying in the more general symmetric cones state space domain, including the interior of the cone $S_d^{+}$ (see also the recent developements in \cite{Cuchiero_Phd}).

The main contribution of this paper consists in relaxing the commutativity assumption made in \cite{article_Bru} and proving that it is possibile to characterize explicitly the joint distribution of the Wishart process and its time integral for a general class of (even not symmetric) mean-reversion and diffusion matrices satisfying the assumptions above with a general constant drift term $b$. The proof of our general Cameron Martin formula is in line with that of theorem 2'' in Bru and we will provide a step-by-step derivation. The study of transform formulae for affine processes on $S_d^+$ is also the topic of \cite{kk2013}, where results concerning Wishart bridges are also provided.

The paper is organized as follows: in section 2 we prove our main result, which extends the original approach by Bru. In section 3 we recall some other existing methods which have been employed in the past literature for the computation of the Laplace transform: the variation of constant, the linearization and the Runge-Kutta method. The first two methods provide analytical solutions, so they should be considered as \textit{competitors} of our new methodology. We show that the variation of constants method is unfeasible for real-life computations, hence the truly analytic competitor is the linearization procedure. After that, we present some applications of our methodology to various settings: a multifactor stochastic volatiliy model, a stochastic correlation model, a short rate model and finally we present a new approach for the computation of a solution to the Algebraic Riccati equation.

\section{The Matrix Cameron-Martin Formula}
\subsection{The Wishart process from the point of view of affine processes}
Before we introduce our result, we would like to report some notations and terminology from \cite{article_Cuchiero}. Let us recall first the definition of an affine process.

\begin{definition}\label{Affine_Process_Definition}
 A Markov process $S$ on $S_d^+$ is called \emph{affine} if it is stochastically continuous and its Laplace transform has exponential-affine dependence on the initial state, i.e. the following equation holds for all $t \geq 0$ and $u \in S_{d}^{+}$:
    \begin{equation}
      \mathbb{E}\left[e^{-Tr\left[u S_t\right]}\right] = \int_{S_{d}^{+}}\! {e^{-\operatorname{Tr}\left[u \xi \right]}\, p_{t}\!\left(x,d\xi\right)} = e^{-\phi\left(t,u\right) - \operatorname{Tr}\left[ \psi\left(t,u\right) S_0\right]}\,,
    \end{equation}
     for some functions $\phi: \mathbb{R}_{+} \times S_{d}^{+} \rightarrow \mathbb{R}_{+}$ and $\psi: \mathbb{R}_{+}\times S_{d}^{+} \rightarrow S_{d}^{+}$.
\end{definition}

In \cite{article_Cuchiero}, a complete characterization of affine processes on $S_d^+$ is provided in terms of the so-called admissible parameter set (see Definition 2.3 in \cite{article_Cuchiero}), which constitutes the affine analogue of a L\' evy triplet. The Wishart process with dynamics \eqref{Wis_dyn} is a conservative pure diffusion affine process with admissible parameter set $\left(\alpha,b,B(x),0,0,0,0\right)$, where $B(x)=Mx+xM^\top$ and $\alpha = Q^\top Q$. Since the process is affine, it is possible to reduce the Kolmogorov PDE associated to the computation of (\ref{Laplace}) to a non linear (matrix Riccati) ODE. 
\begin{proposition}\label{prop_1}
(\cite{article_Cuchiero}) Let $S_t\in WIS_d(S_0, b, M, Q)$ be the Wishart process defined by (\ref{Wis_dyn}), then 
\begin{align*}
\mathbb{E}_{S_0}^{\mathbb{P}}\left[\exp\left\{-Tr\left[wS_t+\int_{0}^{t}{vS_sds}\right]\right\}\right]=\exp\left\{-{\phi}(t)-Tr\left[{\psi}(t)S_0\right]\right\},
\end{align*}
where the functions ${\psi}$ and ${\phi}$ satisfy the following system of ODE's.
\begin{align}
\frac{d {\psi}}{d t}&={\psi}M+M^{\top}{\psi}-2{\psi}Q^{\top}Q{\psi}+v\quad{\psi}(0)=w,\label{ODEpsi}\\
\frac{d {\phi}}{d t}&=Tr\left[b{\psi}(t)\right]\quad{\phi}(0)=0.\label{ODEphi}
\end{align}
\end{proposition}

\subsection{Statement of the result}
In this section we proceed to prove the main result of this paper. We report a formula completely in line with the Matrix Cameron-Martin formula given by \cite{article_Bru}. 
 
\begin{theorem}\label{main}
Let $S\in WIS_d(S_0, b, M,Q)$ be the Wishart process solving (\ref{Wis_dyn}), assume
\begin{align}
M^{\top}\left(Q^\top Q\right)^{-1}=\left(Q^\top Q\right)^{-1}M,\label{cond_true_m}
\end{align}
let $b\succeq (d+1)Q^\top Q$ and define the set of convergence of the Laplace transform 
\begin{align*}
\mathcal{D}_t=&\left\{ w,v\in S_d : \mathbb{E}_{S_0}^{\mathbb{P}}\left[\exp\left\{-Tr\left[wS_t+\int_{0}^{t}{vS_sds}\right]\right\}\right]<+\infty\right\}.\end{align*}
Then for all $u,v\in \mathcal{D}_t$ the joint moment generating function of the process and its integral is given by:
\begin{align*}
&\mathbb{E}_{s_0}^{\mathbb{P}}\left[\exp\left\{-Tr\left[wS_t+\int_{0}^{t}{vS_sds}\right]\right\}\right]=\exp\left\{-\phi(t)-Tr\left[\psi(t)s_0\right]\right\},
\end{align*}
where the functions $\phi$ and $\psi$ are given by:
\begin{align*}
\psi(t)&=\frac{\left(Q^\top Q\right)^{-1}M}{2}-\frac{Q^{-1}\sqrt{\bar{v}}k(t)Q^{\top^{-1}}}{2},\nonumber\\
\phi(t)&=Tr\left[b\frac{\left(Q^\top Q\right)^{-1}M}{2}\right]t\nonumber\\
&+\frac{1}{2}Tr\left[\left(Q^\top\right)^{-1}b\left(Q\right)^{-1}\log\left(\sqrt{\bar{v}}^{-1}\left(\sqrt{\bar{v}}\cosh(\sqrt{\bar{v}}t)+\bar{w}\sinh(\sqrt{\bar{v}}t)\right)\right)\right],
\end{align*}
with $k(t)$ given by: 
\begin{align*}
k(t)=-\left(\sqrt{\bar{v}}\cosh(\sqrt{\bar{v}}t)+\bar{w}\sinh(\sqrt{\bar{v}}t)\right)^{-1}\left(\sqrt{\bar{v}}\sinh(\sqrt{\bar{v}}t)+\bar{w}\cosh(\sqrt{\bar{v}}t)\right)
\end{align*}
and $\bar{v},\bar{w}$ are defined as follows:
\begin{align*}
\bar{v}&=Q\left(2v+M^\top Q^{-1}Q^{\top^{-1}}M\right)Q^\top,\\
\bar{w}&=Q\left(2w-\left(Q^\top Q\right)^{-1}M\right)Q^\top.
\end{align*}
Moreover, the set where the Laplace transform is regular at least contains the area defined by 
\begin{align}
v &\succ -M^\top(2Q^\top Q)^{-1} M \label{conditionv},\\
w& \succeq (2Q^\top Q)^{-1}M-Q^{-1}\sqrt{\bar{v}}(2Q^\top)^{-1}.\label{conditionw}
\end{align}
\end{theorem}
%

\begin{remark} The derivation of Theorem \ref{main} involves a change of probability measure that will be illustrated in the sequel. This change of measure introduces a lack of symmetry which does not allow to derive a fully general formula. However, under the assumption \eqref{cond_true_m} we are able to span a large class of processes. 
In fact, the equality \eqref{cond_true_m} requires the symmetry of a matrix: this involves $d(d-1)/2$ linear equalities in $d^2$ variables, therefore if we fix the parameters of the matrix $(Q^T Q)^{-1}$ and consider the constraints on the parameters of $M$, we get $d^2- d(d-1)/2=d(d+1)/2$ degrees of freedom for choosing the matrix $M$.

In the two dimensional case, let:
\begin{align*}
\left(Q^\top Q\right)^{-1}=\left(\begin{array}[pos]{cc}
a & b\\
b & c	
\end{array}\right),\quad M=\left(
\begin{array}[pos]{cc}
x & y\\
z & t	
\end{array}\right),
\end{align*}
then condition \eqref{cond_true_m} can be expressed as:
\begin{align*}
bx+cz=ay+tb,
\end{align*}
meaning that we can span a large class of parameters, thus going far beyond the commutativity assumption $QM=MQ$ for $Q\in S_d,M\in S_d^-$ as in \cite{article_Bru}. 
\end{remark}
\begin{remark}
Conditions (\ref{conditionv}) - (\ref{conditionw}) give explicit parameter constraints in order to ensure the finiteness of the Laplace transform. If they are not satisfied, then the Laplace transform is regular only up to a (possibly finite) explosion time. Notice that conditions (\ref{conditionv}), (\ref{conditionw}) extend the usual assumptions $v,w\in S_d^+$ as in \cite{article_Cuchiero}.
\end{remark}

\subsection{Proof of Theorem \ref{main}}

We will prove the theorem in several steps. We first consider a simple Wishart process with $M=0$ and $Q=I_d$, defined under a measure $\tilde{\mathbb{P}}$ equivalent to $\mathbb{P}$. The second step will be given by the introduction of the volatility matrix $Q$, using an invariance result. Finally, we will prove the extension for the full process by relying on a measure change from $\tilde{\mathbb{P}}$ to $\mathbb{P}$. Under this last measure, the Wishart process will be defined by the dynamics \eqref{Wis_dyn}.\\

As a starting point we fix a probability measure $\tilde{\mathbb{P}}$ such that $\tilde{\mathbb{P}}\approx\mathbb{P}$. Under the measure $\tilde{\mathbb{P}}$ we consider a matrix Brownian motion $\hat{B}=(\hat{B}_t)_{t\geq0}$, which allows us to define the process $\Sigma_t\in WIS_d(S_0, \tilde{b}, 0, I_d)$, i.e. a process that solves the following matrix SDE:
\begin{align}
d\Sigma_t=\sqrt{\Sigma_t}d\hat{B}_t+d\hat{B}_t^{\top}\sqrt{\Sigma_t}+\tilde{b} dt,\quad \Sigma_0\in S^+_d,\label{Wishartbasic}
\end{align}
where the drift term $\tilde{b}$ satisfies the following condition:
\begin{align*}
\tilde{b}-\left(d+1\right)I_d\in S^+_d.
\end{align*}

For this process, relying on \cite{article_PY} and \cite{article_Bru}, we are able to calculate the Cameron-Martin formula. For the sake of completeness we report the result in \cite{article_Bru}, which constitutes an extension of the methodology introduced in \cite{article_PY}. The result was proved for the restrictive drift $\alpha I_d$, but we will extend it to the general drift by looking to the system of Riccati ODEs.
 
\begin{proposition}\label{propBru} (\cite{article_Bru} Proposition 5 p.742) If $\mathbf{\Phi}:\mathbb{R}_{+}\rightarrow S_d^+$ is continuous, constant on $\left[t,\infty\right[$ and such that its right derivative (in the distribution sense) $\mathbf{\Phi}'_d:\mathbb{R}_{+}\rightarrow S_d^-$ is continuous, with $\mathbf{\Phi}_d(0)=I_d$, and $\mathbf{\Phi}'_d(t)=0$, then for every Wishart process $\Sigma_t\in WIS_d(\Sigma_0, \alpha, 0, I_d)$ we have:
\begin{align*}
\mathbb{E}\left[\exp\left\{-\frac{1}{2}Tr\left[\int_{0}^{t}{\mathbf{\Phi}''_{d}(s)\mathbf{\Phi}^{-1}_{d}(s)\Sigma_sds}\right]\right\}\right]=\left(\det{\mathbf{\Phi}}_d(t)\right)^{\alpha/2}\exp\left\{\frac{1}{2}Tr\left[\Sigma_0\mathbf{\Phi}_{d}^{+}(0)\right]\right\},
\end{align*}
where
\begin{align*}
\mathbf{\Phi}_{d}^{+}(0):=\lim_{t\searrow 0}\mathbf{\Phi}'_{d}(t).
\end{align*}
\end{proposition}

We employ this result to prove the following claim, which establishes the Cameron Martin formula for the more general drift $\tilde{b}$.

\begin{proposition}
Let $\Sigma\in WIS_d(\Sigma_0, \tilde{b}, 0, I_d)$, then
\begin{align}
\mathbb{E}\left[\exp\left\{-\frac{1}{2}Tr\left[w\Sigma_t+\int_{0}^{t}{v\Sigma_s ds}\right]\right\}\right]&=\exp\left\{-{\phi}(t)-Tr\left[{\psi}(t)\Sigma_0\right]\right\},
\end{align}
where
\begin{align*}
\psi(t)&=-\frac{\sqrt{v}k(t)}{2}\\
\phi(t)&=\frac{1}{2}Tr\left[\tilde{b}\log\left(\sqrt{{v}}^{-1}\left(\sqrt{{v}}\cosh(\sqrt{{v}}t)+{w}\sinh(\sqrt{{v}}t)\right)\right)\right]
\end{align*}
and $k(t)$ is given by
\begin{align*}
k(t)=-\left(\sqrt{{v}}\cosh(\sqrt{{v}}t)+{w}\sinh(\sqrt{{v}}t)\right)^{-1}\left(\sqrt{{v}}\sinh(\sqrt{{v}}t)+w\cosh(\sqrt{{v}}t)\right).
\end{align*}
\end{proposition}

\proof
Let us first assume that $\tilde{b}=\alpha I_d$. An application of Proposition \ref{propBru} allows us to claim that
\begin{align}
\mathbb{E}\left[\exp\left\{-\frac{1}{2}Tr\left[w\Sigma_t+\int_{0}^{t}{v\Sigma_s ds}\right]\right\}\right]&=\det\left(\cosh\left(\sqrt{v}t\right)+\sinh\left(\sqrt{v}t\right)k(t)\right)^{\frac{\alpha}{2}}\nonumber\\
&\times\exp\left\{\frac{1}{2}Tr\left[\Sigma_0\sqrt{v}k(t)\right]\right\},\label{PitYor}
\end{align}
where $k(t)$ is given by
\begin{align*}
k(t)=-\left(\sqrt{{v}}\cosh(\sqrt{{v}}t)+{w}\sinh(\sqrt{{v}}t)\right)^{-1}\left(\sqrt{{v}}\sinh(\sqrt{{v}}t)+w\cosh(\sqrt{{v}}t)\right).
\end{align*}
A direct inspection of \eqref{PitYor}, allows us to recognize the functions $\phi$ and $\psi$ in this setting. For $\psi$ we have
\begin{equation}
\psi(t) = -\frac{\sqrt{v}k(t)}{2} 
\end{equation}
which is independent of $\tilde{b}$. The corresponding system of matrix Riccati ODE is
\begin{align}
\frac{d {\psi}}{d t}&=-2{\psi}{\psi}+v\quad{\psi}(0)=w,\\
\frac{d {\phi}}{d t}&=Tr\left[\tilde{b}{\psi}(t)\right]\quad{\phi}(0)=0.
\end{align}
Given the solution for $\psi$, we can determine an alternative formulation for $\phi$ upon integration. This alternative formulation encompasses the more general constant drift too. We show the calculation in detail:
\begin{align*}
\frac{d \phi}{d t}&=Tr\left[\tilde{b}\psi(t)\right]\nonumber\\
&=Tr\left[\tilde{b}\left(-\frac{\sqrt{v}k(t)}{2}\right)\right].\nonumber
\end{align*}
Integrating the ODE yields
\begin{align*}
\phi(t)&=-\frac{1}{2}Tr\left[\tilde{b}\sqrt{{v}}\int_{0}^{t}k(s)ds\right].
\end{align*}
We concentrate on the integral appearing in the second term:
\begin{align*}
&\int_{0}^{t}k(s)ds
&=\int_{0}^{t}{-\left(\sqrt{{v}}\cosh(\sqrt{{v}}s)+{w}\sinh(\sqrt{{v}}s)\right)^{-1}\left(\sqrt{{v}}\sinh(\sqrt{{v}}s)+{w}\cosh(\sqrt{{v}}s)\right)ds}.
\end{align*}
Define $f(s)=\sqrt{{v}}\cosh(\sqrt{{v}}s)+{w}\sinh(\sqrt{{v}}s)$ and let us differentiate it:
\begin{align*}
\frac{d f}{d s}=\left(\sqrt{{v}}\sinh(\sqrt{{v}}s)+{w}\cosh(\sqrt{{v}}s)\right)\sqrt{{v}},
\end{align*}
hence we can write
\begin{align*}
\phi(t)&=\frac{1}{2}Tr\left[\tilde{b}\left(\log\left(\sqrt{{v}}\cosh(\sqrt{{v}}t)+{w}\sinh(\sqrt{{v}}t)\right)-\log\left(\sqrt{{v}}\right)\right)\right]\nonumber\\
&=\frac{1}{2}Tr\left[\tilde{b}\log\left(\sqrt{{v}}^{-1}\left(\sqrt{{v}}\cosh(\sqrt{{v}}t)+{w}\sinh(\sqrt{{v}}t)\right)\right)\right].
\end{align*}
\endproof

\paragraph{\textbf{Invariance under transformations.}}

We define the transformation $S_t=Q^{\top}\Sigma_t Q$, which is governed by the SDE:
\begin{align}
dS_t=\sqrt{S_t}d\tilde{B}_tQ+Q^{\top}d\tilde{B}_t^{\top}\sqrt{S_t}+b dt,\quad b=Q^\top \tilde{b}Q,
\label{WishartQ}\end{align}
where the process $\tilde{B}=(\tilde{B}_t)_{t\geq0}$ defined by $d\tilde{B}_t=\sqrt{S_t}^{-1}Q^{\top}\sqrt{\Sigma}d\hat{B}_t$ is easily proved to be a Brownian motion under $\tilde{\mathbb{P}}$. 


From \cite{article_Bru}, we know how to extend the Cameron Martin formula: the Laplace transform of the process $S$ may be computed as follows 
\begin{align*}
\mathbb{E}_{S_0}^{\tilde{\mathbb{P}}}\left[e^{-Tr\left[wS_t\right]}\right]&=\mathbb{E}^{\tilde{\mathbb{P}}}_{\left(Q^{\top}\right)^{-1}S_0Q^{-1}}\left[e^{-Tr\left[wQ^{\top}\Sigma Q\right]}\right]\nonumber\\
&=\mathbb{E}_{\Sigma_0}^{\tilde{\mathbb{P}}}\left[e^{-Tr\left[\left(QwQ^{\top}\right)\Sigma\right]}\right],
\end{align*}
hence we can compute the Cameron Martin formula for the process $S$ using the arguments $QwQ^\top$ and $QvQ^\top$.

\paragraph{\textbf{Inclusion of the drift - Girsanov transformation.}} The final step consists in introducing a measure change from $\tilde{\mathbb{P}}$, where the process has no mean reversion, to the measure $\mathbb{P}$ that will allow us to consider the general process governed by the dynamics in equation \eqref{Wis_dyn}. We now define a matrix Brownian motion under the probability measure $\mathbb{P}$ as follows:
\begin{align*}
B_t=\tilde{B}_t-\int_{0}^{t}{\sqrt{S_s}M^\top Q^{-1}ds}=\tilde{B}_t-\int_{0}^{t}{H_sds},
\end{align*}
for $H_s=\sqrt{S_s}M^\top Q^{-1}$. The Girsanov transformation is given by the following stochastic exponential (see e.g. \cite{matsu04}):
\begin{align*}
\left.\frac{\partial \mathbb{P}}{\partial \tilde{\mathbb{P}}}\right|_{\mathcal{F}_t}&=\exp\left\{\int_{0}^{t}{Tr\left[H^{\top}d\tilde{B}_s\right]}-\frac{1}{2}\int_{0}^{t}{Tr\left[HH^{\top}\right]ds}\right\}\nonumber\\
&=\exp\left\{\int_{0}^{t}{Tr\left[Q^{-1^{\top}}M\sqrt{S_s}d\tilde{B}_s\right]}-\frac{1}{2}\int_{0}^{t}{Tr\left[S_sM^{\top}Q^{-1}Q^{-1^{\top}}M\right]ds}\right\}.\nonumber
\end{align*}
We concentrate on the stochastic integral term, which under the parametric restriction \eqref{cond_true_m}, can be expressed as
\begin{align*}
\frac{1}{2}\int_{0}^{t}{Tr\left[\left(Q^\top Q\right)^{-1}M\left(\sqrt{S_s}d\tilde{B}_sQ+Q^\top d\tilde{B}_s^\top\sqrt{S_s}\right)\right]}=\frac{1}{2}\int_{0}^{t}{Tr\left[\left(Q^\top Q\right)^{-1}M\left(dS_s-bds\right)\right]}.
\end{align*}
In summary, the stochastic exponential may be written as
\begin{align*}
\left.\frac{\partial \mathbb{P}}{\partial \tilde{\mathbb{P}}}\right|_{\mathcal{F}_t}&=\exp\left\{Tr\left[\frac{\left(Q^\top Q\right)^{-1}M}{2}\left(S_t-S_0-bt\right)\right]-\frac{1}{2}\int_{0}^{t}{Tr\left[S_sM^{\top}Q^{-1}Q^{-1^{\top}}M\right]ds}\right\}.
\end{align*}
Under the assumption $b \succeq (d+1)Q^\top Q$ (which is a sufficient condition ensuring that the process does not hit the boundary of the cone $S_d^+$, see Corollary 3.2 in \cite{article_MPS}), using the same arguments as in \cite{Mayer2012} shows that the stochastic exponential is a true martingale.

\paragraph{\textbf{Derivation of the Matrix Cameron-Martin formula.}}
We consider the process under $\mathbb{P}$:
\begin{align*}
dS_t=\sqrt{S_t}dB_tQ+Q^{\top}dB_t^{\top}\sqrt{S_t}+\left(MS_t+S_tM^{\top}+b\right) dt.
\end{align*}
%

Recall that under $\tilde{\mathbb{P}}$, we have
\begin{align*}
dS_t=\sqrt{S_t}d\tilde{B}_tQ+Q^{\top}d\tilde{B}_t^{\top}\sqrt{S_t}+b dt,
\end{align*}
then $\Sigma_t=Q^{-1^{\top}}S_tQ^{-1}$ solves
\begin{align*}
d\Sigma_t=\sqrt{\Sigma_t}d\hat{B}_t+d\hat{B}_t^{\top}\sqrt{\Sigma_t}+\tilde{b} dt.
\end{align*}
We are now ready to apply the change of measure along the following steps:
\begin{align*}
\mathbb{E}_{S_0}^{\mathbb{P}}\left[\exp\left\{-\frac{1}{2}Tr\left[wS_t+\int_{0}^{t}{vS_sds}\right]\right\}\right]
=&\mathbb{E}_{S_0}^{\tilde{\mathbb{P}}}\left[\exp\left\{-\frac{1}{2}Tr\left[wS_t+\int_{0}^{t}{vS_sds}\right]\right.\right.\nonumber\\
&+Tr\left[\frac{\left(Q^\top Q\right)^{-1}M}{2}\left(S_t-S_0-bt\right)\right]\nonumber\\
&\left.\left.-\frac{1}{2}\int_{0}^{t}{Tr\left[S_sM^{\top}Q^{-1}Q^{-1^{\top}}M\right]ds}\right\}\right]\nonumber\\
=&\exp\left\{-Tr\left[\frac{\left(Q^\top Q\right)^{-1}M}{2}\left(S_0+bt\right)\right]\right\}\nonumber\\
&\times\mathbb{E}_{S_0}^{\tilde{\mathbb{P}}}\left[\exp\left\{-\frac{1}{2}Tr\left[\left(w-\left(Q^\top Q\right)^{-1}M\right)S_t\right.\right.\right.\nonumber\\
&\left.\left.\left.+\int_{0}^{t}{\left(v+M^{\top}Q^{-1}Q^{-1^{\top}}M\right)S_s ds}\right]\right\}\right].
\end{align*}
But $S_t=Q^{\top}\Sigma_t Q$, then:
\begin{align*}
\mathbb{E}_{S_0}^{\mathbb{P}}\left[\exp\left\{-\frac{1}{2}Tr\left[wS_t+\int_{0}^{t}{vS_sds}\right]\right\}\right]=&
\exp\left\{-Tr\left[\frac{\left(Q^\top Q\right)^{-1}M}{2}\left(S_0+bt\right)\right]\right\}\nonumber\\
&\times\mathbb{E}_{Q^{\top^{-1}}S_0Q^{-1}}^{\tilde{\mathbb{P}}}\left[\exp\left\{-\frac{1}{2}Tr\left[Q\left(w-\left(Q^\top Q\right)^{-1}M\right)Q^{\top}\Sigma_t\right.\right.\right.\nonumber\\
&\left.\left.\left.+\int_{0}^{t}{Q\left(v+M^{\top}Q^{-1}Q^{-1^{\top}}M\right)Q^{\top}\Sigma_s ds}\right]\right\}\right].
\end{align*}
The expectation may be computed via a direct application of formula \eqref{PitYor} and after some standard algebra we get the result of Theorem \ref{main}, with the obvious substitutions $v\rightarrow 2v$ and $w\rightarrow 2w$.

\paragraph{\textbf{Strip of regularity for $\mathcal{D}_t$.}}

Here we show that conditions (\ref{conditionv}) and (\ref{conditionw}) imply the boundedness of the Laplace transform for all $t\geq0$.
By theorem (3.7) of \cite{spreve2010} we know that the Laplace transform exists till the explosion time of the solution of the corresponding Riccati ODE. Knowing the explicit solution of such ODE, a sufficient condition for non explosion is that  for all $t\geq0$
\begin{align*}
h(t)= \sqrt{\bar{v}}\cosh(\sqrt{\bar{v}}t)+\bar{w}\sinh(\sqrt{\bar{v}}t)& \in GL_d.
\end{align*}

As $\bar{v}$ appears in a square root, it must be that $\bar{v}\succeq 0$. However, the inequality must be strict due to $h(0)\in GL_d$, then $\bar{v}\succ 0$, i.e. condition (\ref{conditionv}).
Now let us rewrite $h(t)$ as follows:
\begin{align*}
h(t)&= \sqrt{\bar{v}}\frac{e^{\sqrt{\bar{v}}t}+e^{-\sqrt{\bar{v}}t}}{2}+\bar{w}\frac{e^{\sqrt{\bar{v}}t}-e^{-\sqrt{\bar{v}}t}}{2}\\
&=\sqrt{\bar{v}}e^{-\sqrt{\bar{v}}t}+\frac{1}{2}(\sqrt{\bar{v}}+\bar{w})\left( e^{\sqrt{\bar{v}}t}-e^{-\sqrt{\bar{v}}t}\right).
\end{align*}
		
For $\bar{v}\succ 0$ we have that both $e^{-\sqrt{\bar{v}}t}$ and $e^{\sqrt{\bar{v}}t}-e^{-\sqrt{\bar{v}}t}$ belong to $S_d^+$ for all $t\geq0$, then $h(t)\in GL_d$ if $\sqrt{\bar{v}}+\bar{w}\succeq 0$ , which is condition (\ref{conditionw}).

\section{Alternative existing methods}
\subsection{Variation of Constants Method}

The variation of constants method represents the first solution provided in literature for the solution of the matrix ODE's (\ref{ODEpsi}) - (\ref{ODEphi}) (see e.g. \cite{article_gou01}, \cite{article_gou02}, \cite{article_gou03}) and despite its theoretical simplicity, it turns out to be very time consuming, as we will show later in the numerical exercise. This is also equivalent to the procedure followed by \cite{article_alfonsi} and \cite{Mayer_Wis} who found the Laplace transform of the Wishart process alone (i.e. corresponding to $v=0$ in (\ref{Laplace})). The proof of the following proposition is standard and it is omitted.

\begin{proposition}\label{var_const_f} The solutions for ${\psi}(t),{\phi}(t)$ in Proposition \ref{prop_1} are given by:
\begin{align}
{\psi}(t)&=\psi^{\prime}+e^{\left(M^{\top}-2\psi^{\prime}Q^{\top}Q\right)t}\Bigg[\left(w-\psi^{\prime}\right)^{-1}\Bigg.\nonumber\\
&\Bigg.+2\int_{0}^{t}{e^{\left(M-2Q^{\top}Q\psi^{\prime}\right)s}Q^{\top}Qe^{\left(M^{\top}-2\psi^{\prime}Q^{\top}Q\right)s}ds}\Bigg]^{-1}e^{\left(M-2Q^{\top}Q\psi^{\prime}\right)t},\label{PSI_VC}\\
{\phi}(t)&=Tr\left[\alpha Q^{\top}Q\int_{0}^{t}{{\psi}(s)ds}\right],
\end{align}
where $\psi^{\prime}$ is a symmetric solution to the following algebraic Riccati equation
\begin{align}
\psi^{\prime}M+M^{\top}\psi^{\prime}-2\psi^{\prime}Q^{\top}Q\psi^{\prime}+v=0.\label{algebraic}
\end{align}
\end{proposition}

\subsection{Linearization of the Matrix Riccati ODE}
The second approach we consider is the one proposed by \cite{gra08}, who used the Radon lemma in order to linearize the matrix Riccati ODE (\ref{ODEpsi}) (see also \cite{article_Levin}, \cite{yong99} and \cite{anderson}). 
\begin{proposition}(\cite{gra08})\label{linear_chap1}
The functions ${\psi}(t),{\phi}(t)$ in Proposition \ref{prop_1} are given by
\begin{equation*}
{\psi}(t)=\left(w{\psi}_{12}(t)+{\psi}_{22}(t) \right)^{-1}\left(w{\psi}_{11}(t)+{\psi}_{21}(t) \right),
\end{equation*}
\begin{equation*}
{\phi}(t)=\frac{\alpha}{2}Tr\left[log\left(w{\psi}_{12}(t)+{\psi}_{22}(t) \right)+M^{\top}t \right],
\end{equation*}
where
\begin{equation*}
\left(\begin{array}{rr}
{\psi}_{11}(t) & {\psi}_{12}(t)\\
{\psi}_{21}(t) & {\psi}_{22}(t)
\end{array}\right)
=exp\left\{ t \left( 
\begin{array}{rr}
M & 2Q^{\top}Q\\
v & -M^{\top}
\end{array}
\right) \right\}.
\end{equation*} 
\end{proposition}

\subsection{Runge-Kutta Method}
The Runge-Kutta method is a classical approach for the numerical solution of ODE's. For a detailed treatment, see e.g. \cite{book_QSS}. If we want to solve numerically the system of equations \eqref{ODEpsi} and \eqref{ODEphi}, the most commonly used Runge-Kutta scheme is the fourth order one:
\begin{align*}
{\psi}(t_{n+1})&={\psi}(t_{n})+\frac{1}{6}h\left(k_1+2k_2+2k_3+k_4\right),\nonumber\\
t_{n+1}&=t_n+h,\nonumber\\
k_1&=g(t_n,{\psi}(t_n)),\nonumber\\
k_2&=g(t_n+\frac{1}{2}h,{\psi}(t_n)+\frac{1}{2}hk_1),\nonumber\\
k_3&=g(t_n+\frac{1}{2}h,{\psi}(t_n)+\frac{1}{2}hk_2),\nonumber\\
k_4&=g(t_n+h,{\psi}(t_n)+hk_3),\nonumber
\end{align*}
where the function $g$ is given by:
\begin{align*}
g(t_n,{\psi}(t_n))&=g({\psi}(t_n))={\psi}(t_n)M+M^{\top}{\psi}(t_n)-2{\psi}(t_n)Q^{\top}Q{\psi}(t_n)+v.\nonumber
\end{align*}

\subsection{Comparison of the methods}
A formal numerical analysis of the various methods is beyond the scope of this paper. Anyhow, we would like to stress some points about the execution time, which we believe are sufficient to highlight the importance of our new methodology. Despite its importance in the academic literature, it will turn out that the variation of constants method is not suitable for applications, in particular in a calibration setting. 

First of all we  compare the results of the four different methods. We consider different time horizons $t\in\left[0,3.0\right]$ and use the following values for the parameters:
\begin{align*}
&S_0=\left(\begin{array}{cc} 0.0120&0.0010\\ 0.0010&0.0030\end{array}\right);\quad Q=\left(\begin{array}{cc} 0.141421356237310&-0.070710678118655\\0 &0.070710678118655\end{array}\right);\nonumber\\
&M=\left(\begin{array}{cc} -0.02&-0.02\\ -0.01&-0.02\end{array}\right);\quad\alpha=3;\nonumber\\
&v=\left(\begin{array}{cc} 0.1000&0.0400\\ 0.0400&0.1000\end{array}\right);\quad w=\left(\begin{array}{cc} 0.1100&0.0300\\ 0.0300&0.1100\end{array}\right).
\end{align*}

The value for $Q$ was obtained along the following steps: given a matrix $A\in S_d^+$ such that $AM=M^\top A$, we compute its inverse and let $Q$ be obtained from a Cholesky factorization of this inverted matrix.

Table \ref{tab:NumericalExperiment} shows the value of the moment generating function for different values of the time horizon $t$. The four methods lead to values which are very close to each other, and this constitutes a first test proving that the new methodology produces correct results. Let us now consider another important point issue, namely the execution speed. In order to obtain a good degree of precision for the variation of constants method, we were forced to employ a fine integration grid. This results in a poor performance of this method in terms of speed. In Figure \ref{fig:calculation_time} we compare the time spent by the three analytical methods for the calculation of the moment generating function. As $t$ gets larger, the execution time for the variation of constants method grows exponentially, whereas the time required by the linearization and the new methodology is the same. The Runge-Kutta method is a numerical solution to the problem, so the real competitors of our methodology are the variation of constants and the linearization method. 

Finally, we compared the linearization of the Riccati ODE to the new methodology. In terms of precision and execution speed the two methodologies seem to provide the same performance, up to the fourteenth digit. This shows that, under the parametric restriction of Theorem \ref{main} our methodology represents a valid alternative. The results are illustrated in Table \ref{tab:NumericalExperiment2} up to the maturity $t=100$.
%

\section{Applications}
\subsection{Pricing of derivatives}
The knowledge of the functional form of the Laplace transform represents an important tool for the application of a stochastic model in mathematical finance. In the following, we will provide two examples of asset pricing models whose Laplace transform is of exponentially affine form and such that our previous results may be applied. The first one is the model proposed by \cite{article_DaFonseca1} which describes the evolution of a single asset, whose instantaneous volatility is modelled by means of a Wishart process. The second is the model introduced in \cite{article_DaFonseca2}, where the evolution of a vector of assets is described by a vector-valued SDE where the Wishart process models the instantaneous variance-covariance matrix of the assets. 

\subsubsection{A stochastic volatility model}
In this subsection we consider the model proposed in \cite{article_DaFonseca1} and we derive the explicit Laplace transform of the log-price using our new methodology. As a starting point, we report the dynamics defining the model:
\begin{align*}
\frac{dX_t}{X_t}&=Tr\left[\sqrt{S_t}\left(dW_tR^\top+dB_t\sqrt{I_d-RR^\top}\right)\right],\\
dS_t&=\left(\alpha Q^{\top}Q+MS_t+S_tM^{\top}\right)dt+\sqrt{S_t}dW_tQ+Q^{\top}dW_t^{\top}\sqrt{S_t},
\end{align*}
where $X_t$ denotes the price of the underlying asset, and the Wishart process acts as a multifactor source of stochastic volatility. $W$ and $B$ are independent matrix Brownian motions and the matrix $R$ parametrizes all possible correlation structures preserving the affinity. This model is a generalization of the (multi-)Heston model, see \cite{Heston93} and \cite{ChristHest}, and it offers a very rich structure for the modelization of stochastic volatilities as the factors governing the instantaneous variance are non-trivially correlated. It is easy to see that the log-price $Y$ is given as
\begin{align*}
dY=-\frac{1}{2}Tr\left[S_t\right]dt+Tr\left[\sqrt{S_t}\left(dW_tR^\top+dB_t\sqrt{I_d-RR^\top}\right)\right].
\end{align*}

We are interested in the Laplace transform of the log-price, i.e.
\begin{align*}
\varphi_t(\tau,-\omega)=\mathbb{E}\left[e^{-\omega Y_T}\left|\mathcal{F}_t\right.\right],\quad \tau:=T-t.
\end{align*}
This expectation satisfies a backward Kolmogorov equation, see \cite{article_DaFonseca1} for a detailed derivation. Since the process $S=(S_t)_{0\leq t\leq T}$ is affine, we make a guess of a solution of the form
\begin{align*}
\varphi_t(\tau,-\omega)=\exp\left\{-\omega \ln X_t -{\phi}(\tau)-Tr\left[{\psi}(\tau) S_t\right]\right\}.
\end{align*}
By substituting it into the PDE, we obtain the system of ODE's
\begin{align}
\frac{d {\psi}}{d \tau}&={\psi}\left(M-\omega Q^\top R^\top\right)+\left(M^{\top}-\omega RQ\right){\psi}-2{\psi}Q^{\top}Q{\psi}-\frac{\omega^2+\omega}{2}I_d\label{ODEpsiSV},\\
{\psi}(0)&=0,\\
\frac{d {\phi}}{d \tau}&=Tr\left[\alpha Q^{\top}Q{\psi}(\tau)\right]\label{ODEphiSV},\\
{\phi}(0)&=0.
\end{align}
If we look at the first ODE, we recognize the same structure as in \eqref{ODEpsi}: instead of $M$ and $v$ we have respectively $M-\omega Q^\top R^\top$ and $-\frac{\omega^2+\omega}{2}I_d$. This means that we can rewrite the solution for ${\psi}$ as
\begin{align*}
{\psi}(\tau)&=\frac{\left(Q^\top Q\right)^{-1}\left(M-\omega Q^\top R^\top \right)}{2}-\frac{Q^{-1}\sqrt{\bar{v}}kQ^{\top^{-1}}}{2},\\
{\phi}(\tau)&=-\frac{\alpha}{2}\log\left(\det\left(e^{-\left(M-\omega Q^\top R^\top\right)\tau}\left(\cosh(\sqrt{\bar{v}\tau})+\sinh(\sqrt{\bar{v}\tau})k\right)\right)\right),\nonumber\\
\bar{v}&=Q\left(2\left(-\frac{\omega^2+\omega}{2}I_d\right)+\left(M^\top-\omega RQ\right) Q^{-1}Q^{\top^{-1}}\left(M-\omega Q^\top R^\top \right)\right)Q^\top,\nonumber\\
\bar{w}&=Q\left(-\left(Q^\top Q\right)^{-1}\left(M-\omega Q^\top R^\top \right)\right)Q^\top,\nonumber\\
k&=-\left(\sqrt{\bar{v}}\cosh(\sqrt{\bar{v}}\tau)+\bar{w}\sinh(\sqrt{\bar{v}}\tau)\right)^{-1}\left(\sqrt{\bar{v}}\sinh(\sqrt{\bar{v}}\tau)+\bar{w}\cosh(\sqrt{\bar{v}}\tau)\right).
\end{align*}

Condition \eqref{cond_true_m} in this setting has the following form
\begin{align}
\left(M-\omega Q^\top R^\top\right)^\top\left(Q^\top Q\right)^{-1}=\left(Q^\top Q\right)^{-1}\left(M-\omega Q^\top R^\top\right).
\end{align}
For fixed $\omega$ we can express the condition  above via the following system
\begin{align}
\left\{\begin{array}{l}
M^\top \left(Q^\top Q\right)^{-1}=\left(Q^\top Q\right)^{-1}M,\\
RQ\left(Q^\top Q\right)^{-1}=\left(Q^\top Q\right)^{-1}Q^\top R^\top .	
\end{array}\right. 
\end{align}

\subsubsection{A stochastic correlation model}

In this subsection we consider the model introduced in \cite{article_DaFonseca2}. This model belongs to the class of multi-variate affine volatility models, for which many interesting theoretical results have been presented in \cite{Cuchiero_Phd}. In this framework we consider a vector of prices together with a stochastic variance-covariance matrix. 
\begin{align*}
dX_t&=Diag(X_t)\sqrt{S_t}\left( dW_t\rho +\sqrt{1-\rho^\top \rho}dB_t\right),\nonumber\\
dS_t&=\left(\alpha Q^{\top}Q+MS_t+S_tM^{\top}\right)dt+\sqrt{S_t}dW_tQ+Q^{\top}dW_t^{\top}\sqrt{S_t},
\end{align*}
where now the vector Brownian motion $Z= W_t\rho +\sqrt{1-\rho^\top \rho}B_t$ is correlated with the matrix Brownian motion $W$ through the correlation vector $\rho$. Using the same arguments as before, we compute the joint conditional Laplace transform of the vector of the log-prices $Y_T=\log(X_T)$
\begin{align*}
\varphi_t(\tau,-\omega)=\mathbb{E}\left[e^{-\omega^\top Y_T}\left|\mathcal{F}_t\right.\right],\quad \tau:=T-t.
\end{align*}

The affine property allows us to write the associated system of matrix Riccati ODE's (see \cite{article_DaFonseca2} for more details), which is given as
\begin{align}
\frac{d {\psi}}{d \tau}&={\psi}\left(M-Q^\top \rho \omega^\top\right)+\left(M^{\top}-\omega \rho^\top Q\right){\psi}-2{\psi}Q^{\top}Q{\psi}\nonumber\\
&-\frac{1}{2}\left(\sum_{i=1}^{d}{\omega_i e_{ii}}+\omega^\top\omega\right)I_d,\label{ODEpsiSC}\\
{\psi}(0)&=0,\\
\frac{d {\phi}}{d \tau}&=Tr\left[\alpha Q^{\top}Q{\psi}(\tau)\right]\label{ODEphiSC},\\
{\phi}(0)&=0.
\end{align}
We recognize the same structure as in Equations \eqref{ODEphi} and \eqref{ODEpsi} where instead of $M$ and $v$, we now have $M-Q^\top \rho \omega^\top$ and $-\frac{1}{2}\left(\sum_{i=1}^{d}{\omega_i e_{ii}}+\omega^\top\omega\right)I_d$ respectively. Consequently, we can compute the solution as
\begin{align*}
{\psi}(\tau)&=\frac{\left(Q^\top Q\right)^{-1}\left(M-Q^\top \rho \omega^\top\right)}{2}-\frac{Q^{-1}\sqrt{\bar{v}}kQ^{\top^{-1}}}{2},\\
{\phi}(\tau)&=-\frac{\alpha}{2}\log\left(\det\left(e^{-\left(M-Q^\top \rho \omega^\top\right)\tau}\left(\cosh(\sqrt{\bar{v}\tau})+\sinh(\sqrt{\bar{v}\tau})k\right)\right)\right),\nonumber\\
\bar{v}&=Q\left(2\left(-\frac{1}{2}\left(\sum_{i=1}^{d}{\omega_i e_{ii}}+\omega^\top\omega\right)I_d\right)+\left(M^{\top}-\omega \rho Q\right) Q^{-1}Q^{\top^{-1}}\left(M-Q^\top \rho \omega^\top\right)\right)Q^\top,\nonumber\\
\bar{w}&=Q\left(-\left(Q^\top Q\right)^{-1}\left(M-Q^\top \rho \omega^\top\right)\right)Q^\top,\nonumber\\
k&=-\left(\sqrt{\bar{v}}\cosh(\sqrt{\bar{v}}\tau)+\bar{w}\sinh(\sqrt{\bar{v}}\tau)\right)^{-1}\left(\sqrt{\bar{v}}\sinh(\sqrt{\bar{v}}\tau)+\bar{w}\cosh(\sqrt{\bar{v}}\tau)\right).
\end{align*}
Condition \eqref{cond_true_m} is rephrased in this setting as follows
\begin{align}
\left(M^\top-\omega\rho^\top Q\right)\left(Q^\top Q\right)^{-1}=\left(Q^\top Q\right)^{-1}\left(M-Q^\top\rho\omega^\top\right),
\end{align}
which may be expressed as
\begin{align}
\left\{
\begin{array}{l}
M^\top\left(Q^\top Q\right)^{-1}=\left(Q^\top Q\right)^{-1}M,\\
\omega\rho^\top Q^{\top^{-1}}=Q^{-1}\rho\omega^\top .	
\end{array}
\right.
\end{align}

This means that the two products have to be symmetric matrices.

\subsubsection{A short rate model}
Our methodology for the computation of the Laplace transform may be directly employed to provide a closed form formula for the price of zero coupon bonds when the short rate is driven by a Wishart process. The Wishart short rate model has been studied in \cite{article_gou02}, \cite{gra08}, \cite{article_BCT}, \cite{article_Chiarella} and \cite{gnoatto2012}. The short rate is modeled as
\begin{equation}
r_t=a+Tr\left[vS_t\right],\label{shortRate}
\end{equation}
where $a\in\mathbb{R}_{\geq 0}$, $v$ is a symmetric positive definite matrix and $S=\left(S_t\right)_{t\geq0}$ is the Wishart process. Standard arbitrage arguments allow us to claim that the price of a zero coupon bond at time $t$ with time to maturity $\tau:=T-t$, denoted by $P_t(\tau)$, is given by the following expectation
\begin{align}
P_t(\tau):&=\mathbb{E}\left[e^{-\int_{t}^{T}{a+Tr\left[vX_{u}\right]du}}|\mathcal{F}_t\right]\nonumber\\
&=\exp\left\{-\phi(\tau)-Tr\left[\psi(\tau)X_t\right]\right\},
\end{align}
where the associated ODE are
\begin{equation}
\frac{\partial \phi}{\partial \tau}=Tr\left[\alpha Q^{\top}Q\psi(\tau)\right]+a,\quad \phi(0)=0,\label{ODE1}
\end{equation}
and
\begin{align}
\frac{\partial \psi}{\partial \tau}&=\psi(\tau)M+M^{\top}\psi(\tau)-2\psi(\tau)Q^{\top}Q\psi(\tau)+v, \quad \psi(0)=0.\label{ODE2}
\end{align}
We can employ again the Cameron-Martin formula and write the solution to the system as follows
\begin{align}
{\psi}(\tau)&=\frac{\left(Q^\top Q\right)^{-1}M}{2}-\frac{Q^{-1}\sqrt{\bar{v}}kQ^{\top^{-1}}}{2},\\
{\phi}(\tau)&=-\frac{\alpha}{2}\log\left(\det\left(e^{-M\tau}\left(\cosh(\sqrt{\bar{v}}\tau)+\sinh(\sqrt{\bar{v}}\tau)k\right)\right)\right)+a\tau.
\end{align}

\subsection{A solution to the algebraic Riccati equation \eqref{algebraic}}
As an application  of our result of independent interest we look at the problem of computing a solution to the algebraic Riccati equation (ARE) \eqref{algebraic}. This equation is well known from control theory and only numerical methods are available for computing its solution. We will construct a solution to the ARE by comparing the solution of the system of differential equations in Proposition \ref{prop_1} obtained according to our new methodology and the variation of constant approach.
For convenience, rewrite the system of ODE's \eqref{ODEpsi} \eqref{ODEphi} as follows
\begin{align}
\frac{d {\psi}}{d t}&=\mathcal{R\left(\psi\right)},\quad\psi(0)=w,\\
\frac{d {\phi}}{d t}&=\mathcal{F}\left(\psi\right),\quad\phi(0)=0.
\end{align}
An ARE is given by
\begin{align}
\mathcal{R\left(\psi^\prime\right)}=0.
\end{align}
As before, we denote by $\psi^\prime$ a solution to this equation.

\begin{lemma}\label{lemmaTanh}
Let $O\in S_d^+$, define
\begin{align}
\sinh(O\tau)=\frac{e^{O\tau}-e^{-O\tau}}{2},\quad \cosh(O\tau)=\frac{e^{O\tau}+e^{-O\tau}}{2}
\end{align}
and
\begin{align}
\tanh(O\tau)=\left(\cosh(O\tau)\right)^{-1}\sinh(O\tau),\quad\coth(O\tau)=\left(\sinh(O\tau)\right)^{-1}\cosh(O\tau),
\end{align}
then
\begin{align}
\lim_{\tau\to\infty}{\tanh(O\tau)}=\lim_{\tau\to\infty}{\coth(O\tau)}=I_d.
\end{align}
\end{lemma}
\begin{proof}

Let $A\in M_{d}$. If $\Re(\lambda(A))<0,\quad\forall\lambda\in\sigma(A)$, then it is well known that
\begin{equation}
\lim_{\tau\to\infty}{e^{A\tau}}=0\in M_{d\times d}.
\end{equation}
As a consequence we have that
\begin{align}
\lim_{\tau\to\infty}{\tanh(O\tau)}=\lim_{\tau\to\infty}{\left(I_d+e^{-2O\tau}\right)^{-1}\left(I_d-e^{-2O\tau}\right)}=I_d.\label{limitTanh}
\end{align}
The second equality follows along the same lines.
\end{proof}

Let us recall some well known results from control theory. We refer to the review article by \cite{kucera73}. Let us write $v=C^\top C$. We introduce the following notions.
\begin{itemize}
\item The pair $\left(M,Q\right)$ is said to be stabilizable if $\exists$ a matrix $L$ such that $M+QL$ is stable, i.e. all eigenvalues are negative.
\item The pair $\left(C,M\right)$ is said to be detectable if $\exists$ a matrix $F$ stuch that $FC+M$ is stable.
\end{itemize}
We introduce again the matrix $M-2Q^\top Q\psi^\prime$ and call it the closed loop system matrix. A classical result is the following.

\begin{theorem}\label{existence_uniqueness}
Stabilizability of $\left(M,Q\right)$ and detectability of $\left(C,M\right)$ is necessary and sufficient for the ARE to have a unique non-negative solution which makes the closed loop system matrix stable.
\end{theorem}

Now, looking at the variation of constant approach we can prove the next result.

\begin{corollary}\label{cor_infty}
\begin{align}
\lim_{\tau\to\infty}{{\psi}(\tau)}=\psi^\prime.
\end{align}
\end{corollary}
\begin{proof}
Under the assumptions of Theorem \ref{existence_uniqueness}, we have $\lambda\left(M-2Q^{\top}Q\psi^\prime\right)<0$, $\forall \lambda \in \sigma\left(M-2Q^{\top}Q\psi^\prime\right)$ hence we know that the integral in \eqref{PSI_VC}, the solution for $\psi$, is convergent, moreover we know that $e^{\left(M-2Q^{\top}Q\psi^\prime\right)\tau}\searrow 0$ as $\tau\to\infty$, hence the proof is complete.
\end{proof}

This last corollary tells us that the function $\psi$ tends to a stability point of the Riccati ODE. This allows us to claim that, as $\tau\to\infty$, we have $\mathcal{R}\left(\psi(\tau)\right)\searrow 0$.
A nice consequence of this fact is that we are able to provide a new representation for $\psi^\prime$, which constitutes another application of the Cameron-Martin approach.

\begin{proposition}
The value of $\psi^\prime$ in Corollary \ref{cor_infty} admits the following representation
\begin{align}
\psi^\prime=\frac{Q^{-1}\sqrt{\bar{v}}Q^{\top^{-1}}}{2}+\frac{\left(Q^\top Q\right)^{-1}M}{2}.
\end{align}
\end{proposition}

\begin{proof}
On the basis of Theorem \ref{main}, we want to compute
\begin{align}
\lim_{\tau\to\infty}-\frac{Q^{-1}\sqrt{\bar{v}}k(\tau)Q^{\top^{-1}}}{2}+\frac{\left(Q^\top Q\right)^{-1}M}{2}.
\end{align}
To perform the computation, it is sufficient to calculate
\begin{align}
&\lim_{\tau\to\infty}k(\tau)\nonumber\\
&=\lim_{\tau\to\infty}-\left(\sqrt{\bar{v}}\cosh(\sqrt{\bar{v}}\tau)+\bar{w}\sinh(\sqrt{\bar{v}}\tau)\right)^{-1}\left(\sqrt{\bar{v}}\sinh(\sqrt{\bar{v}}\tau)+\bar{w}\cosh(\sqrt{\bar{v}}\tau)\right)\nonumber\\
&=\lim_{\tau\to\infty}-\left(\cosh\sqrt{\bar{v}}\tau\right)^{-1}\left(\sqrt{\bar{v}}+\bar{w}\tanh\sqrt{\bar{v}}\tau\right)^{-1}\left(\sqrt{\bar{v}}+\bar{w}\coth\sqrt{\bar{v}}\tau\right)\sinh\sqrt{\bar{v}}\tau.
\end{align}
From Lemma \ref{lemmaTanh} we know that both $\tanh$ and $\coth$ tend to $I_d$ as $\tau\to\infty$, hence we conclude that
\begin{align}
\lim_{\tau\to\infty}k(\tau)=-I_d
\end{align}
and so we obtain the final claim:
\begin{align}
\lim_{\tau\to\infty}{\psi}(\tau)=\frac{Q^{-1}\sqrt{\bar{v}}Q^{\top^{-1}}}{2}+\frac{\left(Q^\top Q\right)^{-1}M}{2}.
\end{align}
Finally, from Corollary \ref{cor_infty}, we know that $\lim_{\tau\to\infty}{\psi}(\tau)=\psi^\prime$, hence the claim.
\end{proof}

\section{Conclusions}

In this paper we derived a new explicit formula for the joint Laplace transform of the Wishart process and its time integral based on the original approach of \cite{article_Bru}. Our methodology leads to a truly explicit formula that does not involve any additional integration (like the highly time consuming variation of constants method) or blocks of matrix exponentials (like the linearization method) at the price of a simple condition on the parameters. We showed some examples of applications in the context of multifactor and multivariate stochastic volatility. Moreover, we provided an explicit solution to the algebraic Riccati ODE that appears in linear-quadratic control theory and for which only numerical schemes are available. We also recall a recent application of our result by \cite{BauerleLi2013} in the portfolio optimization setting of \cite{article_DaFonseca3}. \\


\appendix

\bibliography{biblio}

\begin{thebibliography}{}

\bibitem[{Ahdida, A., Alfonsi, A.}, 2013]{article_alfonsi}
{Ahdida, A., Alfonsi, A.} (2013).
\newblock {Exact and high order discretization schemes for Wishart processes
  and their affine extensions}.
\newblock {\em Ann. App. Probab.}, 23(3):1025--1073.

\bibitem[{Anderson, B.D.O. and Moore, J.B.}, 1971]{anderson}
{Anderson, B.D.O. and Moore, J.B.} (1971).
\newblock {\em {Linear Optimal Control}}.
\newblock Prentice-Hall, first edition.

\bibitem[{Barndorff-Nielsen, O. E. and Stelzer, R}, 2007]{article_BNS02}
{Barndorff-Nielsen, O. E. and Stelzer, R} (2007).
\newblock Positive-definite matrix processes of finite variation.
\newblock {\em Probab. Math. Statist.}, 27(1):3--43.

\bibitem[{B\"{a}uerle, N. and Li, Z.}, 2013]{BauerleLi2013}
{B\"{a}uerle, N. and Li, Z.} (2013).
\newblock {Optimal portfolios for financial markets with Wishart volatility}.
\newblock {\em {J. Appl. Prob., forthcoming}}.

\bibitem[{Bru, M. F.}, 1991]{article_Bru}
{Bru, M. F.} (1991).
\newblock Wishart processes.
\newblock {\em J. Theoret. Probab.}, 4:725--751.

\bibitem[{Buraschi, A. and Cieslak, A. and Trojani, F.}, 2008]{article_BCT}
{Buraschi, A. and Cieslak, A. and Trojani, F.} (2008).
\newblock {Correlation risk and the term structure of interest rates}.
\newblock {\em SSRN eLibrary}.

\bibitem[{Buraschi, A. and Porchia, P. and Trojani, F.}, 2010]{article_BPT}
{Buraschi, A. and Porchia, P. and Trojani, F.} (2010).
\newblock Correlation risk and optimal portfolio choice.
\newblock {\em The Journal of Finance}, 65(1):393--420.

\bibitem[{Chiarella, C. and Hsiao, C. and To, T.}, 2010]{article_Chiarella}
{Chiarella, C. and Hsiao, C. and To, T.} (2010).
\newblock {Risk premia and Wishart term structure models}.
\newblock {\em Working Paper, SSRN eLibrary}.

\bibitem[{Christoffersen, P. and Heston, S. L. and Jacobs, K.},
  2009]{ChristHest}
{Christoffersen, P. and Heston, S. L. and Jacobs, K.} (2009).
\newblock {The shape and term structure of the index option smirk: why
  multifactor stochastic volatility models work so well}.
\newblock {\em Management Science}, 72:1914--1932.

\bibitem[{Cuchiero, C.}, 2011]{Cuchiero_Phd}
{Cuchiero, C.} (2011).
\newblock {\em Affine and polynomial processes}.
\newblock PhD thesis, ETH Z\"urich.

\bibitem[{Cuchiero, C. and Filipovi\'c, D. and Mayerhofer, E. and Teichmann,
  J.}, 2011]{article_Cuchiero}
{Cuchiero, C. and Filipovi\'c, D. and Mayerhofer, E. and Teichmann, J.} (2011).
\newblock Affine processes on positive semidefinite matrices.
\newblock {\em Ann. App. Prob.}, 21(2):397--463.

\bibitem[{Da Fonseca, J. and Grasselli, M. and Ielpo, F.},
  2011]{article_DaFonseca3}
{Da Fonseca, J. and Grasselli, M. and Ielpo, F.} (2011).
\newblock {Hedging (Co)Variance Risk with Variance Swaps}.
\newblock {\em Int. J. Theoretical Appl. Finance}, 14:899--943.

\bibitem[{Da Fonseca, J. and Grasselli, M. and Ielpo, F.},
  2013]{article_DaFonseca4}
{Da Fonseca, J. and Grasselli, M. and Ielpo, F.} (2013).
\newblock {Estimating the Wishart Affine Stochastic Correlation Model Using the
  Empirical Characteristic Function}.
\newblock {\em {Stud. Nonlinear Dynam. Econometrics, forthcoming}}.

\bibitem[{Da Fonseca, J. and Grasselli, M. and Tebaldi, C.},
  2007]{article_DaFonseca2}
{Da Fonseca, J. and Grasselli, M. and Tebaldi, C.} (2007).
\newblock Option pricing when correlations are stochastic: an analytical
  framework.
\newblock {\em Rev. Derivatives Res.}, 10(2):151--180.

\bibitem[{Da Fonseca, J. and Grasselli, M. and Tebaldi, C.},
  2008]{article_DaFonseca1}
{Da Fonseca, J. and Grasselli, M. and Tebaldi, C.} (2008).
\newblock {A multifactor volatility Heston model}.
\newblock {\em Quant. Finance}, 8(6):591--604.

\bibitem[{Da Fonseca, J. and M. Grasselli}, 2011]{dafgra11}
{Da Fonseca, J. and M. Grasselli} (2011).
\newblock {Riding on the smiles}.
\newblock {\em Quant. Finance}, 11:1609--1632.

\bibitem[{Donati-Martin, C. and Doumerc, Y. and Matsumoto, H. and Yor, M.},
  2004]{matsu04}
{Donati-Martin, C. and Doumerc, Y. and Matsumoto, H. and Yor, M.} (2004).
\newblock {Some properties of Wishart process and a matrix extension of the
  Hartman--Watson law}.
\newblock {\em Publ. Res. Inst. Math. Sci.}, (40):1385--1412.

\bibitem[{Duffie, D. and Filipovi\'c, D. and Schachermayer, W.},
  2003]{article_DFS}
{Duffie, D. and Filipovi\'c, D. and Schachermayer, W.} (2003).
\newblock Affine processes and applications in finance.
\newblock {\em Ann. App. Probab.}, 13:984--1053.

\bibitem[{Gnoatto, A.}, 2012]{gnoatto2012}
{Gnoatto, A.} (2012).
\newblock {The Wishart short rate model}.
\newblock {\em Int. J. Theoretical Appl. Finance}, 15(8).

\bibitem[{Gourieroux, C.}, 2006]{article_Gou06}
{Gourieroux, C.} (2006).
\newblock {Continuous time Wishart process for stochastic risk}.
\newblock {\em Econometric Rev.}, 25:177--217.

\bibitem[{Gourieroux, C. and Monfort, A. and Sufana, R.}, 2005]{article_gou01}
{Gourieroux, C. and Monfort, A. and Sufana, R.} (2005).
\newblock International money and stock market contingent claims.
\newblock Working Papers 2005-41, Centre de Recherche en Economie et
  Statistique.

\bibitem[{Gourieroux, C. and Sufana, R.}, 2003]{article_gou02}
{Gourieroux, C. and Sufana, R.} (2003).
\newblock {Wishart quadratic term structure models}.
\newblock {\em SSRN eLibrary}.

\bibitem[{Gourieroux, C. and Sufana, R.}, 2005]{article_gou03}
{Gourieroux, C. and Sufana, R.} (2005).
\newblock {Derivative pricing with Wishart multivariate stochastic volatility}.
\newblock {\em J. Bus. Econ. Statist.}, 2004(October):1--44.

\bibitem[{Grasselli, M. and C. Tebaldi}, 2008]{gra08}
{Grasselli, M. and C. Tebaldi} (2008).
\newblock Solvable affine term structure models.
\newblock {\em Math. Finance}, 18:135--153.

\bibitem[{Heston, L. S.}, 1993]{Heston93}
{Heston, L. S.} (1993).
\newblock A closed-form solution for options with stochastic volatility with
  applications to bond and currency options.
\newblock {\em Rev. Finan. Stud.}, 6:327--343.

\bibitem[{Jong, J. and Zhou, X.}, 1999]{yong99}
{Jong, J. and Zhou, X.} (1999).
\newblock {\em Stochastic controls: Hamiltonian systems and HJB equations}.
\newblock Springer, New York, 1st edition.

\bibitem[{Kang, C. and Kang, W.}, 2013]{kk2013}
{Kang, C. and Kang, W.} (2013).
\newblock {Transform formulae for linear functionals of affine processes and
  their bridges on positive semidefinite matrices}.
\newblock {\em {Stoch. Proc. App.}}, 123(6):2419 -- 2445.

\bibitem[{Kucera, V.}, 1973]{kucera73}
{Kucera, V.} (1973).
\newblock {A review of the matrix Riccati equation}.
\newblock {\em {Kybernetika}}, 9(1):42--61.

\bibitem[{Leippold, M. and Trojani, F.}, 2010]{article_LT}
{Leippold, M. and Trojani, F.} (2010).
\newblock {Asset pricing with matrix jump diffusions}.
\newblock {\em SSRN eLibrary}.

\bibitem[{Levin, J. J.}, 1959]{article_Levin}
{Levin, J. J.} (1959).
\newblock {On the matrix Riccati equation}.
\newblock {\em Proc. Amer. Math. Soc.}, 10:519--524.

\bibitem[{Mayerhofer, E.}, 2012]{Mayer2012}
{Mayerhofer, E.} (2012).
\newblock {Wishart processes and Wishart distributions: an affine processes
  point of view}.
\newblock {\em CIMPA lecture notes}.

\bibitem[{Mayerhofer, E.}, 2013]{Mayer_Wis}
{Mayerhofer, E.} (2013).
\newblock {On the existence of non-central Wishart distributions}.
\newblock {\em {J. Multiv. Anal.}}, 114(0):448--456.

\bibitem[{Mayerhofer, E. and Pfaffel, O. and Stelzer, R.}, 2011]{article_MPS}
{Mayerhofer, E. and Pfaffel, O. and Stelzer, R.} (2011).
\newblock {On strong solutions for positive definite jump-diffusions}.
\newblock {\em {Stoch. Proc. App.}}, 121(9):2072--2086.

\bibitem[{Muhle-Karbe, J. and Pfaffel, O. and Stelzer, R.},
  2012]{article_MPS02}
{Muhle-Karbe, J. and Pfaffel, O. and Stelzer, R.} (2012).
\newblock {Option pricing in multivariate stochastic volatility models of OU
  type}.
\newblock {\em SIAM J. Fin. Math.}, 3:66--94.

\bibitem[{Pigorsch, C. and Stelzer, R.}, 2009]{article_PS}
{Pigorsch, C. and Stelzer, R.} (2009).
\newblock {On the definition, stationary distribution and second order
  structure of positive semidefinite Ornstein-Uhlenbeck type processes}.
\newblock {\em Bernoulli}, 15(3):754--773.

\bibitem[{Pitman, J. and Yor, M.}, 1982]{article_PY}
{Pitman, J. and Yor, M.} (1982).
\newblock {A decomposition of Bessel bridges}.
\newblock {\em Probab. Theory Related Fields}, 59:425--457.

\bibitem[{Quarteroni, A. and Sacco, R. and Saleri, F.}, 2000]{book_QSS}
{Quarteroni, A. and Sacco, R. and Saleri, F.} (2000).
\newblock {\em {Numerical Mathematics}}, volume~37 of {\em {Texts in Applied
  Mathematics}}.
\newblock Springer, first edition.

\bibitem[{Revuz, D. and Yor, M.}, 1994]{book_RY}
{Revuz, D. and Yor, M.} (1994).
\newblock {\em Continuous martingales and {B}rownian motion}, volume 293 of
  {\em Fundamental Principles of Mathematical Sciences}.
\newblock Springer-Verlag, Berlin, second edition.

\bibitem[{Spreij, P. and Veerman, E.}, 2010]{spreve2010}
{Spreij, P. and Veerman, E.} (2010).
\newblock {The affine transform formula for affine jump-diffusions with a
  general closed convex state space}.
\newblock {\em ArXiv e-prints 1005.1099}.

\end{thebibliography}
\bibliographystyle{abbrv}
\addcontentsline{toc}{section}{Bibliography}



\begin{table}[H]
	\centering
		\begin{tabular}{ccccc}
		\hline
Time Horizon		&Lin.&C.-M.&Var. Const.&R.-K.\\
		\hline
		\hline
     0&   0.998291461216988&   0.998291461216988&   0.998291461216988&   0.998291461216988\\
   0.1&   0.997303305375919&   0.997303305375919&   0.997306285702955&   0.997271605593416\\
   0.5&   0.992740622447456&   0.992740622447456&   0.992703104707601&   0.992583959952442\\
   1.0&   0.985698139368470&   0.985698139368470&   0.985426551402640&   0.985389882322825\\
   1.5&   0.977224894409802&   0.977224894409802&   0.976522044659620&   0.976770850175581\\
   2.0&   0.967388334051965&   0.967388334051964&   0.966066486500228&   0.966794966212865\\
   2.5&   0.956261597343174&   0.956261597343174&   0.954144681472186&   0.955535938544691\\
   3.0&   0.943922618087738&   0.943922618087738&   0.940848141282233&   0.943072180564490\\	
   		\hline
		\end{tabular}
	\caption{This table visualizes the joint moment generating function of the Wishart process and its time integral for different time horizons $\tau$. All four methods are considered. It should be noted that the variation of constants method requires a very fine integration grid in order to produce precise values that can be compared with the results of the other methods.}
	\label{tab:NumericalExperiment}
\end{table}

	\begin{figure}[H]
		\centering
			\includegraphics[scale=0.40]{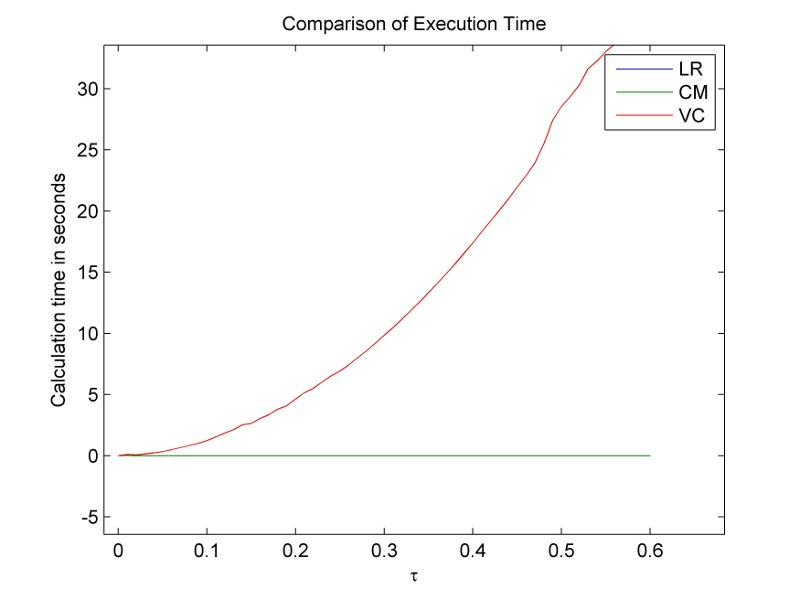}
		
		\caption{In this image we plot the time spent by the three analytical methods (linearization, Cameron-Martin and Variation of Constants) to compute the joint moment generating function of the Wishart process and its time integral for different time horizons. Notice that the Variation of Constants method is inefficient even for quite short maturities.}
	\label{fig:calculation_time}
	\end{figure}
%
		%

\begin{table}[H]
	\centering
		\begin{tabular}{cccc}
\hline
Time horizon&Lin.&C.M.&R.-K.\\
\hline
\hline
     0&   0.998291461216988&   0.998291461216988&   0.998291461216988\\
   0.1&   0.997303305375919&   0.997303305375919&   0.997271605593416\\
   0.5&   0.992740622447456&   0.992740622447456&   0.992583959952442\\
   1.0&   0.985698139368470&   0.985698139368470&   0.985389882322825\\
   2.0&   0.967388334051965&   0.967388334051964&   0.966794966212865\\
   3.0&   0.943922618087738&   0.943922618087738&   0.943072180564490\\
   4.0&   0.915938197508059&   0.915938197508059&   0.914862207389661\\
   5.0&   0.884120166104796&   0.884120166104796&   0.882852196560219\\
  10.0&   0.691634000576684&   0.691634000576684&   0.689897813632122\\
 100.0&   0.000001636282753&   0.000001636282753&   0.000001629036716\\
  \hline
  \end{tabular}
  	\caption{In this table we do not include the results for the variation of constants method. This allows us to look at a longer time horizon and appreciate the precision of the new methodology also in this case.}
	\label{tab:NumericalExperiment2}
\end{table}

\end{document}